\documentclass[conference]{IEEEtran}
\usepackage{amssymb,amsmath,amsthm}
\usepackage{mathrsfs} 

\DeclareFontFamily{OT1}{pzc}{}
\DeclareFontShape{OT1}{pzc}{m}{it}{2 <-> pzcmi8t}{}
\DeclareFontShape{OT1}{pzc}{m}{it}{<-> [1.5] pzcmi8t}{}
\DeclareMathAlphabet{\mathpzc}{OT1}{pzc}{m}{it}

\usepackage{MnSymbol}
\usepackage[latin1]{inputenc}
\usepackage[english]{babel}
\usepackage[]{subfig}

\usepackage[]{color}

\usepackage{epic,eepic,epsfig,graphics}

\usepackage{cite}

\usepackage{color}						
\definecolor{red}{rgb}{1,0,0}	

\newtheorem{theorem}{Theorem}[section]
\newtheorem{lemma}[theorem]{Lemma}

\newtheorem{definition}[theorem]{Definition}

\newcommand{\beq}{\begin{equation}}
\newcommand{\eeq}{\end{equation}}
\newcommand{\bea}{\begin{eqnarray}}
\newcommand{\eea}{\end{eqnarray}}

\newcommand{\be}{\begin{equation}}
\newcommand{\ee}{\end{equation}}
\newcommand{\beqr}{\begin{eqnarray}}
\newcommand{\eeqr}{\end{eqnarray}}
\newcommand{\beqrx}{\begin{eqnarray*}}
\newcommand{\eeqrx}{\end{eqnarray*}}
\newcommand{\ba}{\left[ \begin{array}}
\newcommand{\ea}{\\ \end{array} \right]}
\newcommand{\bi}{\begin{itemize}}
\newcommand{\ei}{\end{itemize}}

\def\xb{{\bf x}}
\def\sb{{\bf s}}
\def\yb{{\bf y}}

\def\vb{{\bf v}}

\def\hb{{\bf h}}

\def\Hb{{\bf H}}

\def\Rbb{\mathbb{R}}
\def\Nc{\mathcal{N}}

\def\Ec{\mathcal{E}}

\def\SNR{\mbox{SNR}}

\def\oneb{\text{\boldmath $1$}}

\begin{document}

\title{On the Mixing Time of Markov Chain Monte Carlo for Integer Least-Square Problems}



\author{
    Weiyu Xu\\
    ECE, University of Iowa\\
    weiyu-xu@uiowa.edu
  \and
    Georgios Alexandros Dimakis\\
    EE,USC\\
    dimakis@usc.edu
    \and
    Babak Hassibi\\
    EE, Caltech\\
    hassibi@systems.caltech.edu
}

\maketitle

\begin{abstract}
In this paper, we study the mixing time of Markov Chain Monte Carlo (MCMC) for integer least-square (LS) optimization problems. It is found that the mixing time of MCMC for integer LS problems depends on the structure of the underlying lattice. More specifically, the mixing time of MCMC is closely related to whether there is a local minimum in the lattice structure. For some lattices, the mixing time of the Markov chain is independent of the signal-to-noise ($SNR$) ratio and grows polynomially in the problem dimension; while for some lattices, the mixing time grows unboundedly as $SNR$ grows. Both theoretical and empirical results suggest that to ensure fast mixing, the temperature for MCMC should often grow positively as the $SNR$ increases. We also derive the probability that there exist local minima in an integer least-square problem, which can be as high as $\frac{1}{3}-\frac{1}{\sqrt{5}}+\frac{2\arctan(\sqrt{\frac{5}{3}})}{\sqrt{5}\pi}$.

\end{abstract}

\section{Introduction}
\label{sec:Introduction}

The integer least-square problem is an NP-hard optimization problem which has received attention in many research areas, for example, communications, global navigation satellite systems, radar imaging, Monte Carlo second-moment estimation, bioinformatics and lattice design \cite{Agrell_et_al_02, Borno}. A computationally efficient way of exactly solving the integer LS problem is the sphere decoder (SD) \cite{Damen_et_al, Hochwald_Ten-Brink_03, Hassibi_1, Agrell_et_al_02}. It is known that for a moderate problem size and a suitable range of Signal-to-Noise Ratios ($SNR$), SD has low computational complexity, which can be significantly smaller than an exhaustive search solver. But for a large problem size and fixed $SNR$, the average computational complexity of SD is still exponential in the problem dimension\cite{Ottersten_05}. So for large problem sizes, (for example large-scale Multiple-Input Multiple-Output (MIMO) systems with many transmit and receive antennas), SD still has high computational complexity and is thus computationally infeasible.

Unlike SD, MCMC algorithms perform a random walk over the signal space in the hope of finding the optimal solution. Gibbs sampling (or Glauber dynamics) is a popular MCMC method which performs the random walk according to the transition probability determined by the stationary distribution of a reversible Markov chain \cite{Levin} \cite{Haggstrom_02}. The Gibbs sampler has been proposed for detection purposes in wireless communication \cite{Zhu_Farhang_Boroujeny_05, Wang_Poor_03} (see also the references therein). These MCMC methods are able to provide the optimal solution if they are run for a sufficiently long time; and empirically MCMC methods are observed to provide near-optimal solutions in a reasonable amount of computational time even for large problem dimensions \cite{Zhu_Farhang_Boroujeny_05, Wang_Poor_03, Hassibi_Globecom}.  \cite{Hassibi_Globecom} gave a characterization of the MCMC temperature parameter such that the optimal solution can be found in polynomial time assuming stationary distribution has been reached. However, the understanding of the mixing time (or the convergence rate, namely how fast a Markov chain converges to the stationary distribution) of these MCMC methods is still limited\cite{Hassibi_Globecom, Farhang_Boroujeny_06, ChenRong}.


 In this paper, we are interested in deriving the mixing time of the Gibbs sampler for integer LS problems. We derive upper and lower bounds on the mixing time and show how the mixing time is related to the structures of integer LS problems. Our work furthers the understanding of the mixing time in MCMC for integer LS problems, and is helpful in optimizing the MCMC parameter for better computational performance.

  Our paper is organized as follows. In Section \ref{sec:System_model} we present the system model. The MCMC method and related background knowledge are introduced in Section \ref{sec:Gibbs_sampling}. Section \ref{sec:mixing_time_ortho}, \ref{sec:mixing_time_local},\ref{sec:local_minimum} and \ref{sec:choice_alpha} derive the bounds on the mixing time and discuss how to optimize MCMC parameters to ensure fast mixing. Simulation results are given in Section \ref{sec:sim_results}. Section \ref{sec:Conclusion} concludes this paper. 
\section{System Model}
\label{sec:System_model}
 In this paper, we consider a real-valued integer least-square problem with $N$ transmit and $N$ receive dimensions, targeting applications in block-fading MIMO antenna systems with known channel coefficients.  The received signal $\yb \in \Rbb^N$ can be expressed as
\begin{equation}
\label{EQ:sig_model_matrix}
\yb = \sqrt{\frac{\SNR}{N}}{\Hb{\xb}} + \vb \ ,
\end{equation}
where ${\xb} \in \Omega^{N}$ is the transmitted signal, and $\Omega$ denotes the constellation set. To simplify the derivations in the paper we will assume that $\Omega = \left\{\pm 1\right\}$. $\vb \in \Rbb^{N}$ is the noise vector where each entry is Gaussian $\Nc \left(0,1\right)$ and independent identically distributed (i.i.d.), and $\Hb \in \Rbb^{N \times N}$ denotes the channel matrix with i.i.d. $\Nc \left(0,1\right)$ entries. The signal-to-noise ratio is defined as
\begin{equation}
\label{EQ:SNR}
\begin{split}
\SNR &= \frac{\Ec \left\|\sqrt{\frac{\SNR}{N}}\Hb{\xb}\right\|^2}{\Ec \|\vb\|^2} \ ,
\end{split}
\end{equation}
which is done in order to take into account the total transmit energy.
Without loss of generality, we assume that the all minus one vector was transmitted, ${\xb} = -\oneb$. Therefore
\begin{equation}
\yb = \vb-\sqrt{\frac{\SNR}{N}}\Hb\oneb \ .
\end{equation}

 To minimize the average error probability, we need to perform Maximum Likelihood Sequence Detection (here simply referred to as ML detection) given by
\begin{equation}
\label{EQ:Original_minimization_problem}
 {\xb}^* = \arg	\mathop {\text{min} }\limits_{{\xb} \in \Omega^{N}} \ \  \left\| \yb - \sqrt{\frac{\SNR}{N}}\Hb {\xb} \right\|^2,
\end{equation}
which is exactly an integer LS problem.
\section{Gibbs Sampling and Mixing Time}
\label{sec:Gibbs_sampling}
In this paper, we investigate one kind of MCMC detector called Gibbs sampler which follows a reversible Markov chain and asymptotically converges to the stationary distribution \cite{Mackay_03}. Under the stationary distribution, the Gibbs sampler has a certain probability of visiting the optimal solution. So if run for sufficiently long time, the Gibbs sampler will be able to find the optimal solution to \eqref{EQ:Original_minimization_problem}.


More specifically, the Gibbs sampler starts with a certain $N$-dimensional feasible vector $\hat{\xb}^{(0)}$ among the set $\{-1,+1\}^{N}$ of cardinality $2^{N}$. Then the Gibbs sampler performs a random walk over $\{-1,+1\}^{N}$ based on the following reversible Markov chain. Assume that we are at time index $l$ and the current state of the Markov chain is $\hat{\xb}^{(l)} \in \{-1,+1\}^{N}$.  In the next step, the Markov chain uniform randomly picks one position index $j$ out of $\{1,2, ..., N\}$ and keeps the symbols of $\hat{\xb}^{(l)}$ at other positions fixed. Then the Gibbs sampler computes the conditional probability of transferring to each constellation point at the $j$-th index. With the symbols at the $(N-1)$ other positions fixed, the probability that the $j$-th symbol adopts the value $\omega$, is given by
\begin{equation}
\label{Eq:Prob_of_symbol_MCMC}
p\left( {\hat{\xb}_j^{(l+1)} = \omega \left|{ \theta }\right.} \right) =  \frac{e^{-\frac{1}{2\alpha^2} \left\| \yb - \sqrt{\frac{\SNR}{N}} \Hb \hat{\xb}_{j \left|{\omega}\right.} \right\|^2   }}{ \sum\limits_{\hat{\xb}_{j \left|{\tilde{\omega}}\right.} \in \Omega}{e^{-\frac{1}{2\alpha^2} \left\| \yb - \sqrt{\frac{\SNR}{N}} \Hb \hat{\xb}_{j \left|{ \tilde{\omega} }\right.} \right\|^2 } }}  \ ,
\end{equation}
where $\tilde{\xb}_{j \left|{\omega}\right. }^T \triangleq \left[\hat{\xb}_{1:j-1}^{(l)}, \omega, \hat{\xb}_{j+1:N}^{(l)} \right]^T$ and $\theta = \left\{ \hat{\xb}^{(l)}, j, \yb, \Hb \right\}$. So conditioned on the $j$-th position is chosen, the Gibbs sampler will with probability $p\left( {\hat{\xb}_j^{(l+1)} = \omega \left|{\theta}\right.} \right)$ keep $\omega$ at the $j$'th index in estimated symbol vector. The initialization of the symbol vector $\hat{\xb}^{(0)}$ can either be chosen randomly or other heuristic solutions. $\alpha$ represents a tunable positive parameter which controls the mixing time of the Markov chain, this parameter is also sometimes called the ``temperature". The smaller $\alpha$ is, the larger the stationary probability for the optimal solution will be, and the easier for the Gibbs sampler to find the optimal solution in the stationary distribution.  But as we will show in the paper, there is often a lower bound on $\alpha$, in order to ensure the fast mixing of the Markov chain to the stationary distribution.

It is not hard to see that the Markov chain of Gibbs sampler is reversible and has $2^{N}$ states with the stationary distribution $e^{-\frac{1}{2\alpha^2} \left\| \yb - \sqrt{\frac{\SNR}{N}} \Hb \hat{\xb} \right\|^2   }$ for an state $\hat{\xb}$. The $2^N \times 2^N$ transition matrix is denoted by $P$, and the element $P_{i,j}$ in the $i$-th ( $1\leq i \leq N$) row and $j$-th ( $1\leq j \leq N$) column is the probability of transferring to state $j$ conditioned on the previous state is $i$. So each row of $P$ sums up to $1$ and the transition matrix after $t$ iterations is $P^{t}$. Denoting the vector for the stationary distribution as $\mathbf{\pi}$, then for an $\epsilon>0$, the mixing time $t(\epsilon)$ is a parameter describing how long it takes for the Markov chain to get close to the stationary distribution, namely, 
\begin{equation*}
t_{mix}(\epsilon):=\min\{t: \max_{\tilde{\xb}} \|P^{t}(\tilde{\xb},\cdot)-\mathbf{\pi}\|_{TV}\},
\end{equation*}
where $\|\mu-\nu\|_{TV}$ is the usual total variation distance between two distributions $\mu$ and $\nu$ over the state space $\{+1,-1\}^{N}$. 
\begin{equation*}
\|\mu-\nu\|_{TV}=\frac{1}{2}\sum_{\yb \in \{+1,-1\}^{N}} |\mu(\yb)-\nu(\yb)|.
\end{equation*}

The mixing time is closely related to the spectrum of the transition matrix $P$. More precisely, for a reversible Markov chain, its mixing time is generally small when the gap between the largest and the second largest eigenvalue of $P$, namely $1-\lambda_2$, is large. The inverse of this gap $\frac{1}{1-\lambda_2}$ is called the relaxation time for this Markov chain. In the next few sections, we will discuss how the mixing time is related to specific system structures.

\section{Mixing Time without Local Minima}
\label{sec:mixing_time_ortho}
In this section, we consider the mixing time for MCMC for integer LS problems and study how the mixing time for integer LS problem depends on the linear matrix structure and $SNR$. As a first step, we consider a linear matrix $\Hb$ with orthogonal columns. As shown later, the mixing time for this matrix has an upper bound independent of $SNR$. In fact, this is a general phenomenon for integer LS problems without local minima.

For simplicity, we incorporate the $SNR$ term into $\Hb$, and the model we are currently considering is
\begin{equation}
\label{eq:orthogonal}
\yb=\Hb \xb +\vb,
\end{equation}
where the columns of $\Hb$ are orthogonal to each other. We will also incorporate the $SNR$ term into $\Hb$ this way in the following sections unless stated otherwise. 

\begin{theorem}
Independent of the temperature $\alpha$ and $SNR$, the mixing time of the Gibbs sampler for orthogonal-column integer least-square problems is upper bounded by $N \log(N)+\log(1/\epsilon)N$.
\end{theorem}

This theorem is an extension of the mixing time for regular random walks on an $N$-dimensional hypercube \cite{Levin}. The only difference here is that the transition probability follows (\ref{Eq:Prob_of_symbol_MCMC}) and that the transition probability depends on $SNR$.

\begin{proof}
When the $k$-th index was selected to update in the Gibbs sampler, since the columns of $\Hb$ are orthogonal to each other, the probability of updating $\xb_{k}$ to $-1$ is $\frac{1}{1+e^{\frac{2\yb^{T} {\bf{h}}_{k}}{\alpha^2}}}$. We note that this probability is independent of the current state of Markov chain $\hat{\xb}$. So we can use the classical coupling idea to get an upper bound on the mixing time of this Markov Chain.

Consider two separate Markov chains starting at two different states $\xb_{1}$ and $\xb_{2}$. These two chains follow the same update rule according to (\ref{Eq:Prob_of_symbol_MCMC}) and, by using the random source, each step they select the same position index to update and they update that position to the same symbol. Let $\tau_{couple}$ be the first time the two chains come to the same state. Then by a classical result, the total variation distance
\begin{equation}
\label{Eq:TV_coupling}
d(t)=\max_{\tilde{\xb}} \|P^{t}(\tilde{\xb},\cdot)-\mathbf{\pi}\|_{TV} \leq \max_{\xb_{1},\xb_{2}} p_{\xb_1,\xb2} \{\tau_{couple}>t\}.
\end{equation}
Note that the coupling time is just time for collecting all of the positions where $\xb_1$ and $\xb_2$ differ, as in the coupon collector problem. From the coupon collector results, for any $\xb_1$ and $\xb_2$,
\begin{equation}
\label{Eq:TV_coupling2}
d(N \log(N)+cN)\leq  p_{\xb_1,\xb_2} \{\tau_{couple}>N \log(N)+cN \} \leq e^{-c}.
\end{equation}
So the conclusion follows.
\end{proof} 
\section{Mixing Time with local Minima}
\label{sec:mixing_time_local}
In this section, we consider the mixing time for integer LS problems which have local minima besides the global minimum point. 

\begin{definition}
A local minimum $\tilde{\xb}$ is a state such that $\tilde{\xb}$ is not a global minimizer for $\min_{\sb \in \{-1,+1\}^N}\|\yb-\Hb\sb\|^2$; and any of its neighbors which differ from $\tilde{\xb}$ in only one position index, denoted by $\tilde{\xb}'$, satisfies $\|\yb-\Hb\tilde{\xb}'\|^2>\|\yb-\Hb\tilde{\xb}\|^2$.
\end{definition}

We will use the following theorem about the spectral gap of Markov chain to evaluate the mixing time.
\begin{theorem}[Jerrum and Sinclair 1989 \cite{JerrumSinclair}, Lawler and Sokal (1988) \cite{Lawler}, \cite{Levin}] Let $\lambda_2$ be the second largest eigenvalue of a reversible transition matrix $P$, and let $\gamma=1-\lambda_2$. Then
\begin{equation*}
                \frac{\Phi_{*}^2}{2}\leq \gamma \leq 2\Phi_{*},
\end{equation*}
where $\Phi_{*}$ is the bottleneck ratio defined as
\begin{equation*}
\Phi_{*}=\min_{\pi(S) \leq \frac{1}{2}} \frac{Q(S,S^{c})}{\pi(S)}.
\end{equation*}
Here $S$ is any subset of the state spaces with stationary measure no bigger than $\frac{1}{2}$, $S^{c}$ is its complement set, and $Q(S,S^c)$ is the probability of moving from $S$ to $S^c$ in one step when starting with the stationary distribution.
\label{thm:gap_bottle}
\end{theorem}

\begin{theorem}
If there is a local minimum ${\tilde{\xb}}$ in an integer least-square problem and we denote its neighbor differing only at the $k$-th ($1\leq k \leq N$) location as ${\tilde{\xb}}_k$, then the mixing time of the Gibbs sampler is at least
\begin{equation}
t_{mix}(\epsilon) \geq \log(\frac{1}{2\epsilon})(\frac{1}{\gamma}-1),
\end{equation}
where
\begin{equation}
\gamma=\sum_{k=1}^{N} \frac{2}{N} {\frac{e^{-\frac{\|\yb-\Hb{\tilde{\xb}}_{k}\|^2}{2\alpha^2}}}{e^{-\frac{\|\yb-\Hb{\tilde{\xb}}_{k}\|^2}{2\alpha^2}}+e^{-\frac{\|\yb-\Hb{\tilde{\xb}}\|^2}{{2\alpha^2}}}} }
\end{equation}

The parameter $\gamma$ is upper bounded by
\begin{equation}
\frac{2}{1+e^{\frac{\min_{k}{\|\yb-\Hb{\tilde{\xb}}_{k}\|^2}-\|\yb-\Hb{\tilde{\xb}}\|^2}{2\alpha^2}}}
\end{equation}
\label{thm:gap_local}
\end{theorem}

\begin{proof}
We apply Theorem \ref{thm:gap_bottle} to prove this result. We take a local minimum point ${\tilde{\xb}}$ as the single element in the bottle-neck set $S$. Since ${\tilde{\xb}}$ is a local minimum, $\pi(S) \leq \frac{1}{2}$.
\begin{equation}
Q(S,S^{c})=\frac{\pi(S)}{N}\sum_{k=1}^{N} {\frac{e^{-\frac{\|\yb-\Hb{\tilde{\xb}}_{k}\|^2}{2\alpha^2}}}{e^{-\frac{\|\yb-\Hb{\tilde{\xb}}_{k}\|^2}{2\alpha^2}}+e^{-\frac{\|\yb-\Hb{\tilde{\xb}}\|^2}{{2\alpha^2}}}} }
\end{equation}

Dividing by $\pi(S)$, by the definition of $\Phi_{*}$
\begin{equation}
\Phi_{*}\leq \frac{Q(S,S^{c})}{\pi(S)}=\frac{1}{N}\sum_{k=1}^{N} {\frac{e^{-\frac{\|\yb-\Hb{\tilde{\xb}}_{k}\|^2}{2\alpha^2}}}{e^{-\frac{\|\yb-\Hb{\tilde{\xb}}_{k}\|^2}{2\alpha^2}}+e^{-\frac{\|\yb-\Hb{\tilde{\xb}}\|^2}{{2\alpha^2}}}} }
\end{equation}

So we know $\gamma \leq 2 \frac{1}{N}\sum_{k=1}^{N} {\frac{e^{-\frac{\|\yb-\Hb{\tilde{\xb}}_{k}\|^2}{2\alpha^2}}}{e^{-\frac{\|\yb-\Hb{\tilde{\xb}}_{k}\|^2}{2\alpha^2}}+e^{-\frac{\|\yb-\Hb{\tilde{\xb}}\|^2}{{2\alpha^2}}}} }$.
From a well-known theorem for the relationship between $t_{mix}(\epsilon)$ and $\gamma$:
$t_{mix}(\epsilon) \geq (\frac{1}{\gamma}-1) \log(\frac{1}{2\epsilon})$ \cite{Levin},
our conclusion follows.
\end{proof}

\begin{theorem}
For an integer least-square problem where no two vectors give the same objective distance, the relaxation time (the inverse of the spectral gap) of MCMC is upper bounded by a constant as the temperature $\alpha \rightarrow 0$ if and only if there is no local minimum. Moreover, when there is a local minimum, as $\alpha \rightarrow 0$, the mixing time of Markov chain $t_{mix}(\epsilon) =e^{\Omega(\frac{1}{2\alpha^2})}$. \footnote{In this paper, $\Omega(\cdot)$,  $\Theta(\cdot)$, and $O(\cdot)$ are the usual scaling notations as in computer science}
\label{thm:mixing_scale_alpha}
\end{theorem}

\begin{proof}
First, when there is a local minimum, from Theorem \ref{thm:gap_local} and Theorem \ref{thm:gap_bottle}, the spectral gap $\gamma$ is lower bounded by
\begin{equation}
\gamma=\frac{2}{N}\sum_{k=1}^{N} {\frac{e^{-\frac{\|\yb-\Hb{\tilde{\xb}}_{k}\|^2}{2\alpha^2}}}{e^{-\frac{\|\yb-\Hb{\tilde{\xb}}_{k}\|^2}{2\alpha^2}}+e^{-\frac{\|\yb-\Hb{\tilde{\xb}}\|^2}{{2\alpha^2}}}} }
\end{equation}

As the temperature $\alpha \rightarrow 0$,  the spectral gap upper bound
\begin{equation}
\frac{2}{1+e^{\frac{\min_{k}{\|\yb-\Hb{\tilde{\xb}}_{k}\|^2}-\|\yb-\Hb{\tilde{\xb}}\|^2}{2\alpha^2}}}
\end{equation}
decreases at the speed of $\Theta(e^{-\frac{\min_{k}{\|\yb-\Hb{\tilde{\xb}}_{k}\|^2}-\|\yb-\Hb{\tilde{\xb}}\|^2}{2\alpha^2}})$. So the relaxation time of the MCMC is lower bounded by $t_{mix}(\epsilon) =e^{\Omega(\frac{1}{2\alpha^2})}$, which grows unbounded as $\alpha \rightarrow 0$.

Suppose instead that there is no local minimum. We argue that as $\alpha \rightarrow 0$, the spectral gap of this MCMC is lower bounded by some constant independent of $\alpha$. Again, we look at the bottle neck ratio and use Theorem \ref{thm:gap_bottle} to bound the spectral gap.

Consider any set  $S$ of sequences which do not include the global minimum point ${\xb^*}$. As $\alpha \rightarrow 0$, the measure of this set of sequences $\pi(S)\leq \frac{1}{2}$. Moreover, as $\alpha \rightarrow 0$, any set $S$ with $\pi(S)\leq \frac{1}{2}$ can not contain the global minimum point ${\xb^*}$. Now we look at the sequence ${\tilde{\xb}}'$ which has the smallest distance $\|\yb-\Hb{\tilde{\xb}}'\|$ among the set $S$. Since there is no local minimum, ${\tilde{\xb}}'$ must have at least one neighbor ${\tilde{\xb}}''$ in $S^{c}$ which has smaller distance than ${\tilde{\xb}}'$. Otherwise, this would imply ${\tilde{\xb}}'$ is a local minimum.  So
\begin{equation}
Q(S,S^{c}) \geq \pi({\tilde{\xb}}') \times \frac{1}{N} {\frac{e^{-\frac{\|\yb-\Hb{\tilde{\xb}}''\|^2}{2\alpha^2}}}{e^{-\frac{\|\yb-\Hb{\tilde{\xb}}''\|^2}{2\alpha^2}}+e^{-\frac{\|\yb-\Hb{\tilde{\xb}}'\|^2}{{2\alpha^2}}}} }
\end{equation}

As $\alpha \rightarrow 0$, $\frac{\pi({\tilde{\xb}}')}{\pi(S)} \rightarrow 1$. So for a given $\epsilon>0$, as $\alpha \rightarrow 0$
\begin{equation}
   \frac{Q(S,S^{c})}{\pi(S)} \geq \frac{1-\epsilon}{N} {\frac{e^{-\frac{\|\yb-\Hb{\tilde{\xb}}''\|^2}{2\alpha^2}}}{e^{-\frac{\|\yb-\Hb{\tilde{\xb}}''\|^2}{2\alpha^2}}+e^{-\frac{\|\yb-\Hb{\tilde{\xb}}'\|^2}{{2\alpha^2}}}} },
\end{equation}
which approaches $\frac{(1-\epsilon)}{N}$ as $\alpha \rightarrow 0$ because $\|\yb-\Hb{\tilde{\xb}}''\|^2 < \|\yb-\Hb{\tilde{\xb}}'\|^2$.

From Theorem \ref{thm:gap_bottle}, the spectral gap $\gamma$ is at least $\frac{(\frac{Q(S,S^{c})}{\pi(S)})^2}{2}$, which is lower bounded by a constant as $\alpha \rightarrow 0$.
\end{proof}

So from the analysis above, the mixing time is closely related to whether there are local minima in the problem. In the next section, we will see there often exist local minima, which implies very slow convergence rate for MCMC when the temperature is kept at the noise level in the high SNR regime. 
\section{The Presence of Local Minima}
\label{sec:local_minimum}
In this section, we look at the problem of how many local minima there are in an integer least-square problem, especially when the $SNR$ is high.

\begin{theorem}
There can be exponentially many local minima in an integer least-quare problem.
\end{theorem}

\begin{proof}
Let $N$ be an even integer. Consider a matrix whose first $\frac{N}{2}$ columns $\hb_{i}$, $1\leq i \leq \frac{N}{2}$ have unit norms and are orthogonal to each other. For the other  $\frac{N}{2}$ columns $\hb_{i}$, $\frac{N}{2} +1\leq i \leq N$, $\hb_{i}=-(1+\epsilon)\hb_{i-\frac{n}{2}}$, where $\epsilon$ is a sufficiently small positive number ($\epsilon<1$). We also let $\yb=\Hb (-\mathbf{{1}})$, where $\mathbf{{1}}$ is an all-$1$ vector. So $-\mathbf{{1}}$ is a globally minimum point for this integer LS problem.

Consider all those vectors ${\tilde{\xb}}'$ which, for any $1\leq i\leq \frac{N}{2}$, its $i$-th element and  $i+\frac{N}{2}$-th element are either simultaneously $+1$ or simultaneously $-1$. When $\epsilon$ is smaller than $1$, we claim that any such a vector except the all $-1$ vector ${\tilde{\xb}}$, is a local minimum, which shows that there are at least $2^{\frac{N}{2}}-1$ local minima.

 Assume that for a certain $1\leq i\leq \frac{N}{2}$, the $i$-th element and  $(i+\frac{N}{2})$-th element of ${\tilde{\xb}}'$ are simultaneously $-1$. Then if we change the $i$-th element to $+1$, $\|\yb-\Hb{\tilde{\xb}}'\|^2$ increases by $4$; and if we change the $(i+\frac{N}{2})$-th element to $+1$, $\|\yb-\Hb{\tilde{\xb}}'\|^2$ increases by $4(1+\epsilon)^2$. This is true because the $i$-th and $(i+\frac{N}{2})$-th columns are orthogonal to other $(N-2)$ columns.

Similarly, assume that for a certain $1\leq i\leq \frac{N}{2}$, the $i$-th element and  $(i+\frac{N}{2})$-th element of ${\tilde{\xb}}'$ are simultaneously $+1$. Then if we change the $i$-th element to $-1$, $\|\yb-\Hb{\tilde{\xb}}'\|^2$ increases by $4(1+\epsilon)^2-4\epsilon^2$; and if we change the $(i+\frac{N}{2})$-th element to $-1$, $\|\yb-\Hb{\tilde{\xb}}'\|^2$ increases by $4-4\epsilon^2$.
\end{proof}

Now we study how often we encounter a local minimum in the specific inter least-square problem model. Without loss of generality, we assume that the transmitted sequence is an all $-1$ sequence. We first give the condition for ${\tilde{\xb}}$ to be a local minimum. We assume that ${\tilde{\xb}}$ is a vector which has $k$ `$+1$' over an index set $K$ with $|K|=k$ and $(N-k)$ `$-1$' over the set $\overline{K}=\{1,2,...,N\}\setminus {K}$.
\begin{lemma}
${\tilde{\xb}}$ is a local minimum if and only if ${\tilde{\xb}}$ is not a global minimum; and
\begin{itemize}
\item  $\forall i \in K$,
\begin{eqnarray}
\hb_{i}^{T}(\sum_{j \in K}\hb_{j}-\frac{\vb}{2})<\frac{\|\hb_{i}\|^2}{2}
\end{eqnarray}
\item  $\forall i \in \overline{K}$,
\begin{eqnarray}
\hb_{i}^{T}(\sum_{j \in K}\hb_{j}-\frac{\vb}{2})>-\frac{\|\hb_{i}\|^2}{2}.
\end{eqnarray}

\end{itemize}

\label{lemma:localcondition}
\end{lemma}

\begin{proof}
For a position $i \in K$, when we flip ${\tilde{\xb}}_{i}$ to $1$,  $\|\yb-\Hb{\tilde{\xb}}'\|^2$ is increased, namely,
\begin{eqnarray}
&&\|\yb-\Hb{\tilde{\xb}}\|^2-\|\yb-\Hb{\tilde{\xb}}_{\sim i}\|^2 \nonumber \\
&=& \|-2\sum_{j \in K}\hb_{j}+\vb\|^2-\|-2\sum_{j \in K, j \neq i}\hb_{j}+\vb\|^2 \nonumber \\
&=& 4\|\hb_{i}\|^2+4\hb_{i}^{T}(2\sum_{j \in K, j \neq i}\hb_{j}-\vb)  \nonumber\\
&<&0,
\end{eqnarray}
where $\tilde{\xb}_{\sim i}$ is a neighbor of $\tilde{\xb}$ by changing index $i$.
This means
\begin{eqnarray}
\hb_{i}^{T}(\sum_{j \in K}\hb_{j}-\frac{\vb}{2})<\frac{\|\hb_{i}\|^2}{2}.
\end{eqnarray}

For a position $i \in \overline{K}$, when we flip ${\tilde{\xb}}_{i}$ to $-1$,  $\|\yb-\Hb{\tilde{\xb}}'\|^2$ is also increased, namely,
\begin{eqnarray}
&&\|\yb-\Hb{\tilde{\xb}}\|^2-\|\yb-\Hb{\tilde{\xb}}_{\sim i}\|^2 \nonumber \\
&=& \|-2\sum_{j \in K}\hb_{j}+\vb\|^2-\|-2\sum_{j \in K}\hb_{j}-2\hb_{i}+\vb\|^2 \nonumber \\
&=& -4\|\hb_{i}\|^2+4\hb_{i}^{T}(-2\sum_{j \in K}\hb_{j}+\vb)  \nonumber\\
&<&0.
\end{eqnarray}
This means
\begin{eqnarray}
(\hb_{i})^{T}(\sum_{j \in K}\hb_{j}-\frac{\vb}{2})>-\frac{\|\hb_{i}\|^2}{2}.
\end{eqnarray}
\end{proof}

%
%

It is not hard to see that when $SNR \rightarrow \infty$, $\vb$ is comparatively small with high probability, so we have the following lemma.
\begin{lemma}
When $SNR \rightarrow \infty$, with high probability, ${\tilde{\xb}}$ is a local minimum if and only if ${\tilde{\xb}} \neq -\mathbf{1}$; and
\begin{itemize}
\item  $\forall i \in K$,
\begin{eqnarray}
\hb_{i}^{T}(\sum_{j \in K}\hb_{j})<\frac{\|\hb_{i}\|^2}{2}
\end{eqnarray}
\item  $\forall i \in \overline{K}$,
\begin{eqnarray}
\hb_{i}^{T}(\sum_{j \in K}\hb_{j})>-\frac{\|\hb_{i}\|^2}{2}.
\end{eqnarray}

\end{itemize}

\end{lemma}

\begin{theorem}
Consider a $2 \times 2$ matrix $\Hb$ whose two columns are uniform randomly sampled from the unit-normed $2$-dimensional vector.
When $\vb=0$, the probability of there existing a local minimum for such an $\Hb$ is $\frac{1}{3}$.
\end{theorem}

\begin{proof}
When $\vb=0$,  clearly ${\tilde{\xb}}=(-1,-1)$ is a global minimum point, not a local minimum point. It is also clear that ${\tilde{\xb}}=(-1,1)$ or ${\tilde{\xb}}=(1,-1)$ can not be a local minimum point since they are neighbors to the global minimum solution. So the only possible local minimum point is ${\tilde{\xb}}=(1,1)$.

From Lemma \ref{lemma:localcondition}, the corresponding necessary and sufficient condition is
\begin{equation*}
\hb_{1}^{T}\hb_{2} < -\frac{\|\hb_{1}\|^2}{2}=-\frac{\|\hb_{2}\|^2}{2}=-\frac{1}{2}.
\end{equation*}
This means the angle $\theta$ between the two 2-dimensional vectors $\hb_{1}$ and $\hb_{2}$ satisfy $\cos(\theta) <-\frac{1}{2}$. Since $\hb_{1}$ and $\hb_{2}$ are two independent uniform randomly sampled vector, the chance for that to happen is $\frac{\pi-\arccos{(-\frac{1}{2})}}{\pi}=\frac{1}{3}$.
\end{proof}

\begin{theorem}
Consider a $2 \times 2$ matrix $\Hb$ whose elements are independent $\Nc (0,1)$ Gaussian random variables.
When $\vb=0$, the probability of there existing a local minimum for such an $\Hb$ is $\frac{1}{3}-\frac{1}{\sqrt{5}}+\frac{2\arctan(\sqrt{\frac{5}{3}})}{\sqrt{5}\pi}$.
\label{thm:22Gaussian}
\end{theorem}

\begin{proof}
When $\vb=0$,  clearly ${\tilde{\xb}}=(-1,-1)$ is a global minimum point, not a local minimum point. It is also clear that ${\tilde{\xb}}=(-1,1)$ or ${\tilde{\xb}}=(1,-1)$ can not be a local minimum point since they are neighbors to the global minimum solution. So the only possible local minimum point is ${\tilde{\xb}}=(1,1)$.

From Lemma \ref{lemma:localcondition}, the corresponding necessary and sufficient condition is
\begin{equation*}
\hb_{1}^{T}\hb_{2} < -\max\left\{ \frac{\|\hb_{1}\|^2}{2}, \frac{\|\hb_{2}\|^2}{2} \right\}.
\end{equation*}
This means the angle $\theta$ between the two 2-dimensional vectors $\hb_{1}$ and $\hb_{2}$ satisfy
\begin{equation*}
r_1 r_2 \cos(\theta) < -\frac{\max\left\{r_1^2, r_2^2\right\}}{2},
\end{equation*}
where $r_1$ and $r_2$ are respectively the $\ell_2$ norm of $\hb_{1}$ and $\hb_{2}$.

 Because the elements of $\Hb$ are independent Gaussian random variables, $r_1$ and $r_2$ are thus independent random variables following the Rayleigh distribution
\begin{eqnarray*}
p(r_1)&=&r_1 e^{-\frac{r_1^2}{2}}\nonumber\\
p(r_2)&=&r_2 e^{-\frac{r_2^2}{2}};
\end{eqnarray*}
while $\theta$ follows a uniform distribution over $[0, 2\pi)$

By symmetry, for $t\geq 1$,
\begin{eqnarray*}
&&P(\frac{\max\left\{r_1^2, r_2^2\right\}}{r_1 r_2}>t) \\
&=&2\int_{0}^{\infty} r_1 e^{-\frac{r_1^2}{2}}  \times {\int_{0}^{\frac{r_1}{t}} r_2 e^{-\frac{r_2^2}{2}} \,dr_2}                           \,dr_1\\
&=&2\int_{0}^{\infty} r_1 e^{-\frac{r_1^2}{2}}  \times {(1-e^{-\frac{r_1^2}{2}})} \,dr_1\\
&=&2(1-\int_{0}^{\infty} r_1 e^{-(\frac{1}{2}+\frac{1}{2t^2})r_1^2} \,dr_1)\\
&=&\frac{2}{t^2+1}.
\end{eqnarray*}

Since $\theta$ is an independent random variable satisfying  $\cos(\theta) < -\frac{\max\left\{r_1^2, r_2^2\right\}}{2r_1 r_2}$ and  $\cos(\theta) \geq -1$, the probability that ${\tilde{\xb}}=(+1,+1)$ is a local minimum is given by
\begin{eqnarray*}
&P=& \int_{1}^{2} (1-\frac{2}{t^2+1})' (1-\frac{\arccos(-\frac{t}{2})}{\pi}) \,dt\\
&=& \int_{1}^{2} \frac{4t}{(t^2+1)^2} (1-\frac{\arccos(-\frac{t}{2})}{\pi}) \,dt. \\
&=&\frac{1}{3}-\frac{1}{\sqrt{5}}+\frac{2\arctan(\sqrt{\frac{5}{3}})}{\sqrt{5}\pi},
\end{eqnarray*}
which is approximately $0.145696$.
\end{proof}

For higher dimension $N$, it is hard to directly estimate the probability of a vector being a local minimum based on the conditions in Lemma \ref{lemma:localcondition}. Simulation results instead suggest that for large $N$, with high probability, there exists at least one local minimum. The following lemma gives us a sufficient condition. For example, if the sum of $k$ columns has a very small $\ell_2$ norm, that will very likely lead to a local minimum.

\begin{lemma}
\begin{eqnarray}
\|\sum_{j \in K}\hb_{j} -\frac{\vb}{2}\| < \min_{i} \frac{\|\hb_{i}\|}{2}.
\end{eqnarray}

\label{lemma:localsufficient}
\end{lemma}

\begin{proof}
This follows from  $|\hb_{i}^T (\sum_{j \in K}\hb_{j} -\frac{\vb}{2})| < \frac{\|\hb_{i}\|^2}{2}$.
\end{proof}

\begin{theorem}
Consider an $N \times N$ matrix $\Hb$ whose $N$ columns are uniform randomly sampled from the unit-normed $N$-dimensional vector.
When $\vb=0$, then the expected number of local minima for such an $\Hb$ is $\Ec(N_{local}) \geq \sum_{k=2}^{N}{\binom{N}{k}}P_{k}$, where $P_{k}$ is the probability that the magnitude of the sum of $k$ uniform randomly sampled vectors is less than $\frac{1}{2}$.
\end{theorem}

\begin{proof}
This follows from Lemma \ref{lemma:localcondition} and the fact that there are $\binom{N}{k}$ vectors for ${\tilde{\xb}}$ which have exactly $k$ +1 in it.
\end{proof}

\section{Choice of Temperature $\alpha$ in High $SNR$}
\label{sec:choice_alpha}
In previous sections, we have looked at the mixing time of MCMC for an integer LS problem. Now we use the results we have accumulated so far to help choose the appropriate temperature of $\alpha$ to ensure that the MCMC mixes fast and that the optimal solution also comes up fast when the system is in a stationary distribution.

When $SNR \rightarrow \infty$, the integer LS problem will have the same local minima as the case $\vb=0$. From the derivations and simulations, it is suggested that with high probability there will be at least one local minimum in the integer LS problem, especially for large problem dimension $N$.

So following from Lemma \ref{thm:mixing_scale_alpha} and the reasoning therein, to ensure there is an upper bound on the mixing time as $SNR\rightarrow \infty$, the temperature $\alpha$ should at least grow at a rate such that
\begin{equation*}
\max_{\tilde{\xb}} \min_{\tilde{\xb}'}\frac{\frac{SNR}{N} \left(\|-\Hb\mathbf{1}-\tilde{\xb}'\|^2-\|-\Hb\mathbf{1}-\tilde{\xb}\|^2\right)}{2\alpha^2} \leq C, 
\end{equation*}
where $\tilde{\xb}$ is a local minimum and $\tilde{\xb}'$ is a neighbor of $\tilde{\xb}$, and $C$ is a constant.

This will require that $\alpha^2$ grow as fast as $\Omega(SNR)$ to ensure fast mixing with the existence of local minima. This explains that if we keep the temperature at the noise level, it will lead to slow convergence in the high SNR regime \cite{Farhang_Boroujeny_06}. 
\section{Simulation Results}
\label{sec:sim_results}
In this section we present simulation results for an $N \times N$ system with a full square channel matrix containing i.i.d. Gaussian entries.

In Figure \ref{fig:N_expected}, we plot the expected number of local minima in a system as the problem dimension $N$ grows. For each $N$, we generate $100$ random channel matrices and for each matrix, we examine the number of local minima by exhaustive search. As the problem dimension $N$ grows, the number of local minima grows rapidly.

In Figure \ref{fig:N_frequency}, we plot the probability of there existing a local minimum as the problem dimension $N$ grows. For each $N$, we generate $100$ random channel matrices and for each matrix, we examine whether there exists local minimum by exhaustive search. As $N$ grows, the empirical probability of there existing at least one local minimum approaches $1$. It is interesting to see that for $N=2$, our theoretical result $\frac{1}{3}-\frac{1}{\sqrt{5}}+\frac{2\arctan(\sqrt{\frac{5}{3}})}{\sqrt{5}\pi}\approx0.15$ matches well with the simulations.

We also examine how the spectral gap for MCMC is related to the existence of local minima. For $N=5$ and $SNR=10$, we randomly generated $10$ problem instances and keep the temperature $\alpha^2=1$ the same as the noise variance. Out of the $10$ trials, the number of local minima are  $2,     1,     0,     0,     0,     2,     0,0,     2$ and  $0$. The corresponding spectral gaps are respectively $0.0037,   0.0008,    0.1244,    0.1957,    0.1989$,    $0.0011$,    $0.1698$,    $0.1764$,    $5 \times 10^{-10}$, and $0.1266$. It can be seen that when there exist local minima, the spectral gap is significantly smaller than the cases without local minima. This implies a slower mixing for the systems with local minima, which is consistent with our theoretical results.

\begin{figure}[tb]
\centering
\includegraphics[width=3.5in, height=2.5in]{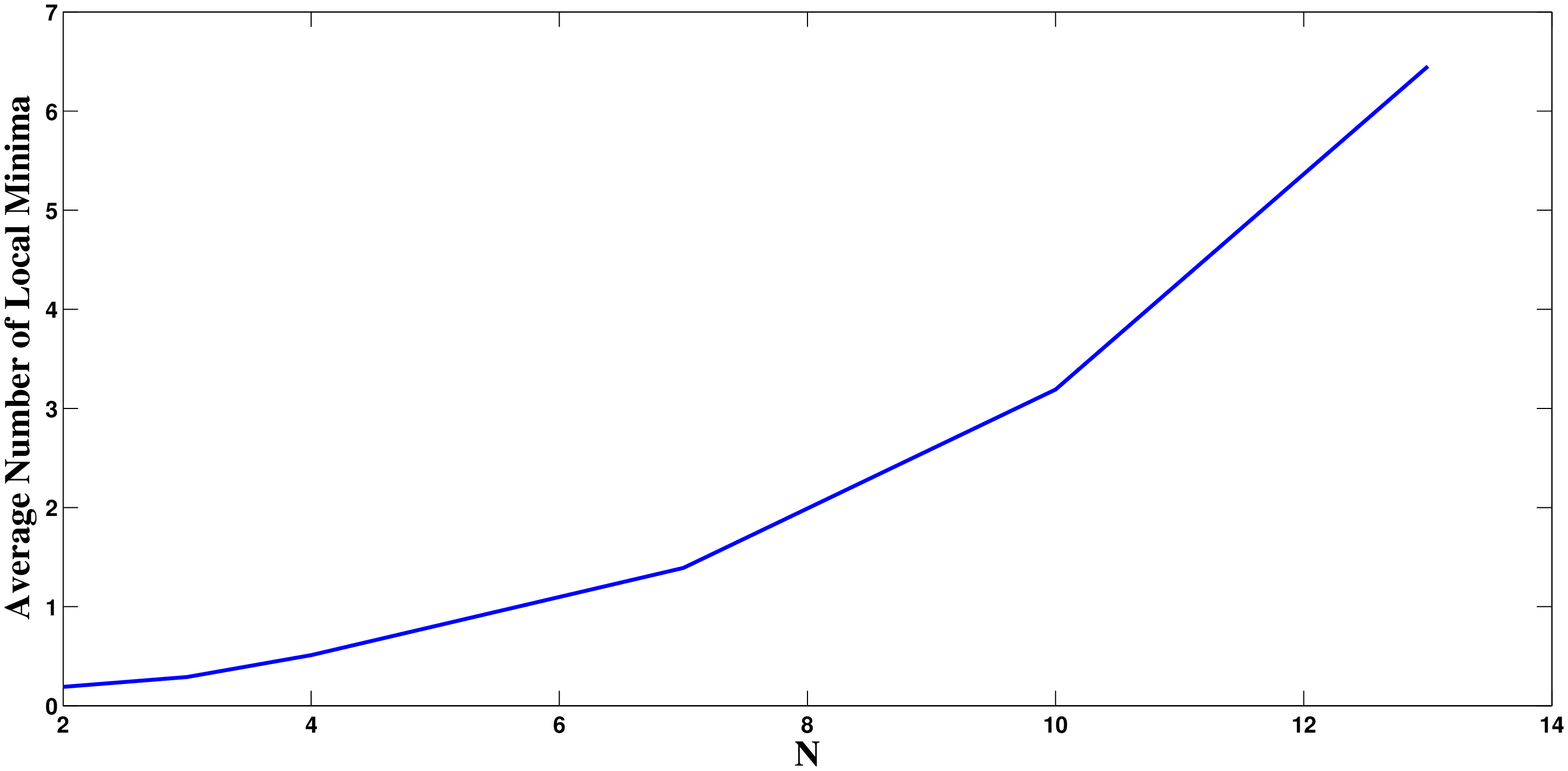}
\caption{Average Number of Local Minima}
\label{fig:N_expected}
\end{figure}
\begin{figure}[tb]
	\centering
  \includegraphics[width=3.5in,height=2.5in]{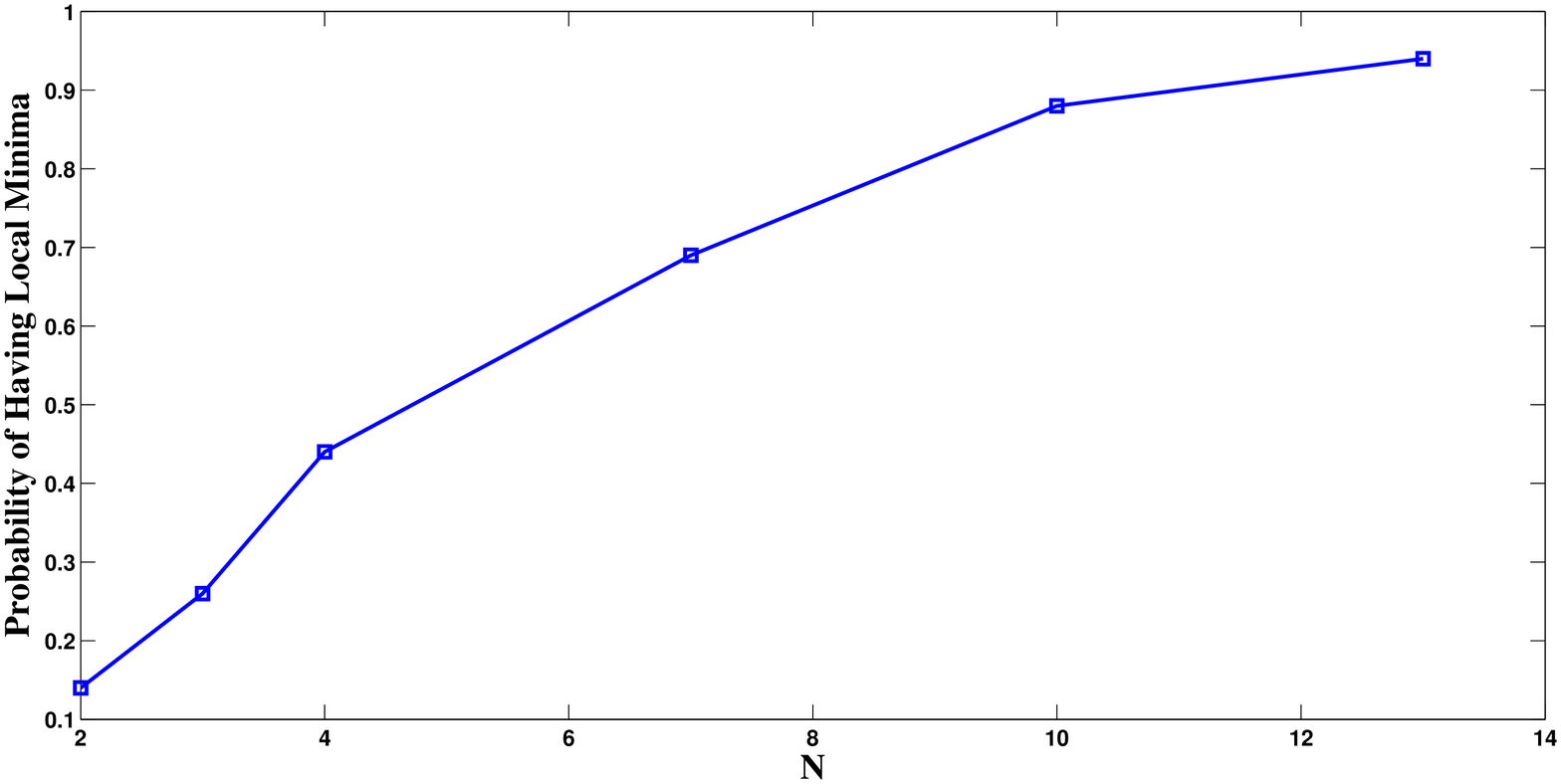}
	\vspace{-2mm}
	\caption{The Probability of Having Local Minima}
	\label{fig:N_frequency}
\end{figure}
\section{Conclusion}
\label{sec:Conclusion}
In this paper, we study the mixing time of Markov Chain Monte Carlo (MCMC) for the integer least-square optimization problem. It is found that the mixing time of MCMC for the integer least-square problem depends on the structure of the underlying lattice. More specifically, the mixing time of MCMC is found to be closely related to whether there is a local minimum in the lattice structure of the integer least-square problem. For some lattices, the mixing time of the Markov chain is independent of the signal-to-noise ratio; while for some lattices, the mixing time is correlated with the signal-to-noise ratio. We also derive the probability that there exist local minima in an integer least-square problem, which can be as high as $\frac{1}{3}-\frac{1}{\sqrt{5}}+\frac{2\arctan(\sqrt{\frac{5}{3}})}{\sqrt{5}\pi}$. Both theoretical and empirical results suggest that to ensure fast mixing for the MCMC for the integer least-square problem, the temperature for MCMC should often grow as the signal-noise-ratio increases. 
\bibliographystyle{IEEEtran}
\bibliography{refs}

\begin{thebibliography}{10}
\providecommand{\url}[1]{#1}
\csname url@samestyle\endcsname
\providecommand{\newblock}{\relax}
\providecommand{\bibinfo}[2]{#2}
\providecommand{\BIBentrySTDinterwordspacing}{\spaceskip=0pt\relax}
\providecommand{\BIBentryALTinterwordstretchfactor}{4}
\providecommand{\BIBentryALTinterwordspacing}{\spaceskip=\fontdimen2\font plus
\BIBentryALTinterwordstretchfactor\fontdimen3\font minus
  \fontdimen4\font\relax}
\providecommand{\BIBforeignlanguage}[2]{{%
\expandafter\ifx\csname l@#1\endcsname\relax
\typeout{** WARNING: IEEEtran.bst: No hyphenation pattern has been}%
\typeout{** loaded for the language `#1'. Using the pattern for}%
\typeout{** the default language instead.}%
\else
\language=\csname l@#1\endcsname
\fi
#2}}
\providecommand{\BIBdecl}{\relax}
\BIBdecl

\bibitem{Agrell_et_al_02}
E.~Agrell, T.~Eriksson, A.~Vardy, and K.~Zeger, ``{Closest point search in
  lattices},'' \emph{IEEE Transactions on Information Theory}, vol.~48, no.~8,
  pp. 2201--2214, 2002.

\bibitem{Borno}
M.~A. Borno, ``{Reduction in Solving Some Integer Least Squares Problems},''
  \emph{Thesis, McGill University}, 2011.

\bibitem{Damen_et_al}
M.~O. Damen, H.~E. Gamal, and G.~Caire, ``{On Maximum-Likelihood Detection and
  the Search for the Closest Lattice Point},'' \emph{IEEE Trans. on Info.
  Theory}, vol.~49, pp. 2389--2402, Oct. 2003.

\bibitem{Hochwald_Ten-Brink_03}
B.~M. Hochwald and S.~Ten~Brink, ``{Achieving near-capacity on a
  multiple-antenna channel},'' \emph{IEEE Trans. on Commun.}, vol.~51, no.~3,
  pp. 389--399, 2003.

\bibitem{Hassibi_1}
B.~Hassibi and H.~Vikalo, ``{On the Sphere-Decoding Algorithm. I. Expected
  Complexity},'' \emph{IEEE Trans. on Sig. Proc.}, vol.~53, pp. 2806--2818,
  Aug. 2005.

\bibitem{Ottersten_05}
J.~Jald\'{e}n and B.~Ottersten, ``{On the Complexity of Sphere Decoding in
  Digital Communications},'' \emph{IEEE Trans. on Sig. Proc.}, vol.~53, pp.
  1474--1484, Apr. 2005.

\bibitem{Levin}
D.~Levin, Y.~Peres, and E.~Wilmer, \emph{Markov Chains and Mixing Times}.\hskip
  1em plus 0.5em minus 0.4em\relax American Mathematical Society, 2008.

\bibitem{Haggstrom_02}
O.~H{\"a}ggstr{\"o}m, \emph{{Finite Markov chains and algorithmic
  applications}}.\hskip 1em plus 0.5em minus 0.4em\relax Cambridge University
  Press, 2002.

\bibitem{Zhu_Farhang_Boroujeny_05}
H.~Zhu, B.~Farhang-Boroujeny, and R.~Chen, ``{On performance of sphere decoding
  and Markov chain Monte Carlo detection methods},'' \emph{IEEE Sig. Proc.
  Letters}, vol.~12, pp. 669--672, 2005.

\bibitem{Wang_Poor_03}
X.~Wang and V.~H. Poor, \emph{Wireless Communications Systems: Advanced
  Techniques for Signal Reception}.\hskip 1em plus 0.5em minus 0.4em\relax
  Prentice Hall, 2003.

\bibitem{Hassibi_Globecom}
M.~Hansen, B.~Hassibi, A.~Dimakis, and W.~Xu, ``{Near-Optimal Detection in MIMO
  Systems using Gibbs Sampling },'' in \emph{Globecom'09}, 2009.

\bibitem{Farhang_Boroujeny_06}
B.~Farhang-Boroujeny, H.~Zhu, and Z.~Shi, ``{Markov chain Monte Carlo
  algorithms for CDMA and MIMO communication systems},'' \emph{IEEE Trans. on
  Sig. Proc.}, vol.~54, no.~5, pp. 1896--1909, 2006.

\bibitem{ChenRong}
R.~Chen, S.~J. Liu, and X.~Wang, ``{Convergence analyses and comparisons of
  Markov chain Monte Carlo algorithms in digital communications},'' \emph{IEEE
  Transactions on Signal Processing}, vol.~50, pp. 255--270, 2002.

\bibitem{Mackay_03}
D.~MacKay, \emph{{Information theory, inference and learning
  algorithms}}.\hskip 1em plus 0.5em minus 0.4em\relax Cambridge University
  Press, 2003.

\bibitem{JerrumSinclair}
M.~Jerrum and A.~Sinclair, ``{Approximating the Permanent},'' \emph{SIAM
  Journal on Computing}, vol.~18, pp. 1149--1178, 1989.

\bibitem{Lawler}
G.~Lawler and A.~Sokal, ``{Bounds on the $L^2$ spectrum for Markov Chains and
  Markov Processes: a Generalization of Cheeger's Inequality},'' \emph{Trans.
  Amer. Math. Soc.}, vol. 309, pp. 557--580, 1988.

\end{thebibliography}

\end{document}